\documentclass[a4paper,twoside,12pt]{scrartcl}
\usepackage[margin=2.5cm]{geometry}

\usepackage[colorlinks=true,linkcolor=blue]{hyperref}%
\expandafter\ifx\csname package@font\endcsname\relax\else
 \expandafter\expandafter
 \expandafter\usepackage
 \expandafter\expandafter
 \expandafter{\csname package@font\endcsname}%
\fi
\usepackage[utf8]{inputenc}
\usepackage[english]{babel}
\usepackage{amsmath}
\usepackage{amsfonts}
\usepackage{amsthm}
\usepackage{graphicx}
\usepackage{mathtools}
\usepackage{physics}

\makeatletter
\newtheoremstyle{indented}
  {3pt}
  {3pt}
  {\itshape\addtolength{\@totalleftmargin}{0.5em}
   \addtolength{\linewidth}{-0.5em}}
  {}
  {\bfseries}
  {.}
  {.5em}
  {}
\makeatother

\newcounter{counter}
\theoremstyle{definition}
\newtheorem{definition}[counter]{Definition}
\theoremstyle{indented}

\newtheorem{theorem}[counter]{Theorem}
\newtheorem{lemma}[counter]{Lemma}
\newtheorem{corollary}[counter]{Corollary}

\newcommand{\half}{\frac{1}{2}}
\newcommand{\sleq}{\sqsubseteq}

\title{Ordering states and channels based on positive Bayesian evidence}
\author{John van de Wetering\\\small{Radboud University}\\\small\texttt{wetering@cs.ru.nl}}

\date{May 2017}

\begin{document}

\maketitle

\begin{abstract}
In this paper we introduce a new partial order on quantum states that considers which states can be achieved from others by updating on `agreeing' Bayesian evidence. We prove that this order can also be interpreted in terms of minimising worst case distuinguishability between states using the concept of quantum max-divergence. This order solves the problem of which states are optimal approximations to their more pure counterparts and it shows in an explicit way that a proposed quantum analogue of Bayes' rule leads to a Bayesian update that changes the state as little as possible when updating on positive evidence. We prove some structural properties of the order, specifically that the order preserves convex mixtures and tensor products of states and that it is a domain. The uniqueness of the order given these properties is discussed. Finally we extend this order on states to one on quantum channels using the Jamio\l{}kowski isomorphism. This order turns the spaces of unital/non-unital trace-preserving quantum channels into domains that, unlike the regular order on channels, is not trivial for unital trace-preserving channels.
\end{abstract}

\section{Introduction}
A common way to get an order of purity on quantum states is to use majorization. It provides a way of looking at purity in terms of which states can be converted into each other \cite{alberti1982stochasticity}. This order does not care about similarity of states: two states can be equivalent in terms of majorization yet their overlap (fidelity) can be zero. If we are interested in which states can be seen as `more pure versions' of other states we need some other order. That is the topic of this paper.

There are quite a few ways to measure similarity and differences for quantum states and channels. Examples include the fidelity, the trace distance or the relative von Neumann entropy. These distances are often based on average case behaviour of states. A candidate for a difference measure based on worst case difference between states and channels is the quantum max-divergence which is attained as a limit of the quantum R\'enyi divergence measures. Operational interpretations of these divergences have been demonstrated \cite{quantumrenyi,muller2013quantum,mosonyi2015quantum}.

We will show that the order in this paper can be understood by looking at what it means to gain `positive' evidence in a Bayesian context: evidence that increases your certainty. In doing so we provide an application of the quantum Bayes' rule used by Fuchs, Leifer and Spekkens \cite{fuchs2002quantum,leifer2013towards}. The order also turns out to be equivalent to minimising the quantum max-divergence between states. This shows that Bayesian updating to positive evidence changes the state as minimally as possible in this sense.

In the theory of computation the notion of a domain is important \cite{scottbook2003, martin2008domain,martinphd}. It provides a framework in which the convergence of computation can be understood. The order we will be studying is a domain structure on the space of quantum states. For qubits this structure coincides with the spectral order for quantum states introduced by Coecke and Martin \cite{coecke2010partial}. The order in this paper is not monotone over all (or just unital) completely positive trace preserving (CPTP) maps. By a simple argument we will show that any `nice' enough partial order on the space of states cannot be monotone over all channels.

The order on states is given by the expression
$$\rho \sleq \sigma \iff \frac{\rho}{\norm{\rho}} \geq \frac{\sigma}{\norm{\sigma}}$$
where $\rho,\sigma\in M_n$ are $n$-dimensional density matrices and $\norm{\rho}$ is the operator norm of the matrix. In section 2 we show that this order establishes the relation of states $\sigma$ that can be obtained from a state $\rho$ by updating your beliefs on `positive' Bayesian evidence using the quantum Bayes' rule advocated by Fuchs, Leifer and Spekkens \cite{fuchs2002quantum,leifer2013towards}:
$$ \rho \sleq \sigma \iff \exists E: \sigma = \frac{\rho^\half E \rho^\half}{\tr(E\rho)} $$
where $E$ is an effect that needs to satisfy a condition related to the positive nature of the evidence. In section 3 we establish some desirable properties of the order. In particular that it is preserved by convex mixtures and tensor products of states, and that it turns the space of states into a domain. We also show a uniqueness result.  An operational interpretation of the order in terms of minimising worst-case (i.e.\ single-shot) distuinguishability is established in section 4:

\noindent\hspace{1em}\begin{minipage}[h!]{\columnwidth-1em}
\vspace{0.5em}
\emph{Suppose we want to construct a state $\sigma$, but we can only achieve a maximum purity of $M\geq H_\infty(\sigma)$ as measured in terms of min-entropy. What state $\rho$ should we construct to minimise the worst case difference between $\sigma$ and $\rho$ as measured in terms of quantum max-divergence?}

\emph{Answer}: Pick a $\rho$ such that $H_\infty(\rho)=M$ and $\rho\sleq \sigma$.
\vspace{0.5em}
\end{minipage}
In section 5 the order is extended to the space of quantum channels. This extension is also a domain structure and is a slight modification of the standard order on von Neumann algebra's \cite{cho2016semantics} with the benefit that it is actually non-trivial on unital channels, but with the drawback that it is no longer monotone over all maps. Finally in section 6 we briefly discuss other orderings of states.

\section{Positive Evidence}
We will start with classical states to gain some intuition. A classical state is simply a probability distribution. Consider the following situation:

We have $n$ boxes and one of them contains some prize money. We have a probability distribution $x$ over these $n$ boxes that represents our knowledge about which of the boxes is likely to contain the prize. A complete lack of knowledge would then be represented by the uniform distribution $x=\bot_n = \frac{1}{n}I_n$ and perfect knowledge would be some pure distribution $P_i$ that is 1 at $i$ and zero everywhere else. If we had to pick a box we would pick the one that we attach the highest probability of winning to. Without loss of generality we take this to be box 1.

Now suppose we gain some evidence, for instance a friend comes along that has played this game many times and knows about the boxes. We would then update our beliefs using Bayes' rule to some new probability distribution. Given that we attach a nonzero prior probability to each box, by updating we could be left with an arbitrary probability distribution after gaining this evidence. In particular we could now be \emph{less} certain about which box to take. However, we are guaranteed to become more certain about our choice if our knowledgeable friend agrees with our prior pick. We will call this kind of evidence \emph{positive evidence}. This is the type of evidence that makes you more certain about your beliefs.

Concretely we can write the prior probabilities as $P(B=i)=x_i$, for which we know that $x_1\geq x_i$ for all $i$. The evidence is $P(E\lvert B=i)=p_i$ and if it is positive evidence than it should agree with our pick so we should have $p_1\geq p_i$. Now our updated probability distribution is
$$
y_i = P(B=i\lvert E) = \frac{P(E\lvert B=i)P(B=i)}{P(E)} = \frac{p_ix_i}{C}
$$
where $C = \sum_i p_ix_i$ is the normalisation constant. Since we gained some positive evidence, $y$ is a more certain (more pure) probability distribution. Now we ask the question, given a probability distribution $x$, what kind of probability distributions $y$ can be reached from $x$ with some positive evidence?

To answer that question we solve $y_i = \frac{p_ix_i}{C}$ assuming that $x_1\geq x_i$ and $p_1\geq p_i$. $C$ is fixed by noting that $C = p_1 \frac{x_1}{y_1}$. Filling it in we then get $y_i = \frac{p_i}{p_1}\frac{x_i}{x_1}y_1$ so that
$$
\frac{y_i}{y_1} = \frac{p_i}{p_1}\frac{x_i}{x_1}.
$$
Because $p_1\geq p_i$ we can find a set of $0\leq p_i\leq 1$ such that this equality holds if and only if 
$$
\frac{y_i}{y_1} \leq \frac{x_i}{x_1}\quad \text{for all } 1\leq i \leq n.
$$
Instead of assuming that $x_1$ is the biggest component we will use $x^+$ to denote the value of the biggest component. This construction now gives us a partial order:
$$
x \sleq y \iff \frac{y_i}{y^+} \leq \frac{x_i}{x^+} \quad \text{for all } 1\leq i \leq n.
$$

Let us now try to generalise this to quantum states. We will modify the game a bit by assuming there is some $n$-dimensional quantum particle. The information that we have about this particle is encoded in some quantum state $\rho\in M_n(\mathbb{C})$, e.g., $\rho$ is a positive operator and trace-normalised: $\tr(\rho)=1$. Now suppose the game is such that we can measure the particle using any pure measurement and we win if we guess the correct outcome. What kind of measurement should we do to maximise the chance of winning? This question boils down to finding a 1-dimensional projector $P$ such that $\tr(P \rho)$ (the chance of observing $P$ via the Born-rule) is maximised. This is the case when $P$ is a projector corresponding to the highest eigenvalue of $\rho$, similarly to the classical case where we had to pick the component with the highest probability. Agreeing `quantum evidence' can therefore be seen as evidence that preserves the highest eigenvalue.

Now comes the question of what we mean by quantum evidence and what the right quantum Bayes' rule is. A natural choice for encoding evidence would be to use efects, which are positive operators $E$ below the identity: $0\leq E\leq 1$. In the classical case we would then write the Bayesian update as $\rho \mapsto \frac{E\rho}{\tr(E\rho)}$, but of course when $E$ and $\rho$ don't commute, this doesn't necessarily result in a positive operator. A solution to this is to `sandwich' the operators, which is known as the sequential product:
$$\rho \mapsto \frac{E^\half\rho E^\half}{\tr(E\rho)}.$$
This version of the quantum Bayes' rule is the standard generalisation of the projection postulate. With this definition of Bayes' rule in hand we can try to create an order on states. Let $L^+(\rho)$ denote the eigenspace corresponding to the highest eigenvalue of $\rho$. Define
\begin{equation*}
 \rho \sleq^\prime \sigma \iff\begin{aligned} &\exists E: \sigma = \frac{E^\half \rho E^\half}{\tr(E\rho)} \\
 &\text{with } L^+(\rho)\cap L^+(E) \neq \{0\}\end{aligned}
 \end{equation*}
where this condition on the subspaces of the highest eigenvalues means that there is a vector $v\neq 0$ such that $\rho v = \rho^+ v$ and $E v = E^+ v$ which we have seen is the quantum analogue of the `agreeing condition' that $p_1\geq p_i$ and $x_1\geq x_i$.

While this relation is reflexive and antisymmetric, transitivity fails. To see this, note that if we can write $\rho^\prime = E^\half\rho E^\half$ and $\rho^{\prime\prime} = F^\half\rho^\prime F^\half$ then there is no guarantee that we can find an effect $G$ such that $\rho^{\prime\prime}$ can be obtained from $\rho$ because the composition $F^\half E^\half$ is not in general an effect. In fact, if we were to take the transitive closure of this relation then we would need to consider all updates of the form $\rho \mapsto \left(U^\dagger E^\half \rho E^\half U\right)/\tr(E\rho)$ where $U$ is an arbitrary unitary that leaves $L^+(\rho)$ intact. This transitive closure is no longer antisymmetric.

There is another proposal for a quantum Bayes' rule that is advocated by, for instance, Fuchs \cite{fuchs2002quantum} and Leifer and Spekkens \cite{leifer2013towards}. They propose that the correct formulation is
$$\rho \mapsto \frac{\rho^\half E \rho^\half}{\tr(\rho E)}.$$
This turns out to be the correct update rule for this problem which we will show after defining the quantum version of the classical order above.
\begin{definition}
Let $\rho,\sigma \in DO(n) = \{A \in M_n(\mathbb{C})~;~A\geq 0, \tr A = 1\}$ be states and denote the highest eigenvalue of $\rho$ as $\rho^+$ (which is just equal to the operator norm $\norm{\rho}$). Define the \emph{quantum positive evidence order} (QPE order) as
$$
 \rho \sleq \sigma \iff \frac{\sigma}{\sigma^+} \leq \frac{\rho}{\rho^+} \iff \sigma^+\rho - \rho^+\sigma \geq 0.
$$
This last expression will be referred to as the \emph{order inequality}.
\end{definition}
\begin{lemma}
\label{lem:basic}
Let $\rho$ and $\sigma$ be states in $M_n$ and suppose that $\rho\sleq \sigma$. Denote the linear subspace of the eigenvectors corresponding to the highest eigenvalue of a state $\rho$ by $L^+(\rho)$, then the following hold.
\begin{enumerate}
\item $\norm{\rho}=\rho^+\leq \sigma^+= \norm{\sigma}$.
\item If $\rho^+=\sigma^+$ then $\rho=\sigma$.
\item $L^+(\sigma)\subseteq L^+(\rho)$.
\item If $\rho v = 0$ then $\sigma v = 0$.
\item $\sleq$ is a partial order: reflexive, transitive and antisymmetric.
\end{enumerate}
\end{lemma}
\begin{proof}
The first two follow by taking the trace of the order inequality. The third follows by enclosing the order inequality by $v^\dagger(\_)v$ for a $v\in L^+(\sigma)$ and the fourth point follows from just plugging a $v$ from the kernel of $\rho$ into the order inequality. Reflexivity and transitivity should be clear from the definition and antisymmetry follows by points 1 and 2.
\end{proof}
\begin{theorem}
\label{theor:evidence}
	Let $\rho$ and $\sigma$ be states in $M_n$.
	$$\sigma\sleq\rho \iff \begin{aligned}&\exists E: \sigma = \frac{\rho^\half E \rho^\half}{\tr(E\rho)}\\
	&\text{with }L^+(E)\cap L^+(\rho) \neq \{0\}.\end{aligned}$$
\end{theorem}
\begin{proof}
	For the `only if' direction let $\sigma = \left(\rho^\half E \rho^\half\right)/\tr(E\rho)$ such that there exists a $v\neq 0$ with $\rho v = \rho^+ v$ and $E v = E^+ v$. Then we note that $\sigma v = \rho^+ E^+/\tr(E\rho) v$ is a maximum eigenvalue of $\sigma$. We need to show that $\sigma^+\rho - \rho^+ \sigma \geq 0$. Filling in the definition we get
	$$\frac{\rho^+ E^+ \rho}{\tr(E\rho)} - \frac{\rho^+ \rho^\half E \rho^\half}{\tr(E\rho)} \geq 0 \iff E^+ \rho - \rho^\half E \rho^\half \geq 0 \iff \rho^\half (E^+I_n - E)\rho^\half \geq 0$$
	and by using that $E\leq E^+ I$ we see that this is true.

	For the other direction we start with $\rho/\rho^+ \geq \sigma/\sigma^+$ and need to find the right $E$. Note that $E$ can be arbitrary outside of the support of $\rho$ so we may restrict to $\rho$'s support for which it has an inverse. Let
	$$E = \left(\frac{\rho}{\rho^+}\right)^{-\half}\frac{\sigma}{\sigma^+}\left(\frac{\rho}{\rho^+}\right)^{-\half}.$$
	Since $\sigma/\sigma^+ \leq \rho/\rho^+$ we indeed have $E\leq 1$. It is easily checked that this $E$ indeed gives $\sigma = \rho^\half E\rho^\half/\tr(E\rho)$.
\end{proof}

So while the standard projection postulate Bayes' rule doesn't define a good partial order, the Fuchs-Leifer-Spekkens (FLS) Bayes' rule does give the correct generalisation to the quantum case for this problem. As Leifer and Spekkens argued \cite{leifer2013towards} there probably won't be a single correct way to do Bayesian statistics in the quantum world, but here we have demonstrated a new problem where this version of the Bayes's rule is the correct generalisation.

\section{The QPE order}
First some general properties of the QPE order.

\begin{theorem}
\label{theor:properties}
Let $\rho$ and $\sigma$ be states in $M_n$. Let $\bot_n = \frac{1}{n}I_n$ denote the completely mixed state on $M_n$.
The following are true for the QPE order:
\begin{enumerate}
\item The completely mixed state is the bottom element: $\bot_n\sleq \rho$ for any $\rho$.
\item The pure states are maximal. A pure state $\ket{v}\bra{v}$ is above $\rho$ iff $v\in L^+(\rho)$.
\item The order is invariant under unitary conjugation: \\ $\rho\sleq \sigma$ iff $U\rho U^\dagger\sleq U\sigma U^\dagger$ for any unitary operator $U\in U(n)$. In fact, the order is invariant under application of any linear (trace preserving) isometry $\Phi:M_n\rightarrow M_k$.
\item When $\rho\sleq \sigma$ we have ker$(\rho)\subseteq $ ker$(\sigma)$ and rnk$(\rho)\geq $ rnk$(\sigma)$.
\item The convex structure of state space is preserved: $\rho\sleq \sigma$ iff $\rho\sleq (1-t)\rho + t\sigma \sleq \sigma$ for all $0\leq t\leq 1$.
\item Downsets are closed convex spaces. Uppersets are also closed and are unions of closed convex spaces\footnote{The QPE order is actually an example of a pospace: The graph induced by the order is closed in the space $DO(n)\times DO(n)$.}.
\item Let $\rho_i,\sigma_i\in DO(n_i)$ with $\rho_i\sleq \sigma_i$ for $i=1,2$, then $\rho_1\otimes \rho_2 \sleq \sigma_1\otimes \sigma_2$.
\item When $\rho\sleq \sigma$ we have $\lambda(\rho)\sleq \lambda(\sigma)$, where $\lambda(\rho)$ denotes the set of ordered eigenvalues of $\rho$.
\end{enumerate}
\end{theorem}
\begin{proof}
We note that $\bot_n^+=\frac{1}{n}$, $(\ket{v}\bra{v})^+=1$ and $\Phi(A)^+=\norm{\Phi(A)} = \norm{A}=A^+$ when $\Phi$ is an isometry. With the help of Lemma \ref{lem:basic}, the first 4 points then follow easily. The fifth points follows because when $\rho\sleq \sigma$ there is a nonzero vector $v$ such that $\sigma v = \sigma^+v$ and $\rho v = \rho^+ v$ by Lemma \ref{lem:basic}. It then follows that $((1-t)\rho + t\sigma)^+= (1-t)\rho^++t\sigma^+$ from which statement 5 follows easily. 

For statement 6, let $\rho_i\sleq\sigma$ for $i=1,2$. Then $L^+(\sigma)\subseteq L^+(\rho_1)\cap L^+(\rho_2)$ from which it follows that $((1-t)\rho_1+t\rho_2)^+ = (1-t)\rho_1^+ + t\rho_2^+$. By writing out the order inequality it is easily seen that $(1-t)\rho_1+t\rho_2\sleq\sigma$, so that downsets are indeed convex sets. When $\rho\sleq\sigma_i$, $i=1,2$ and $L^+(\rho)$ is 1-dimensional we can use the same argument to show that the upperset of $\rho$ is convex. If the dimension of $L^+(\rho)$ is bigger than 1 then each normalised $v\in L^+(\rho)$ corresponds to some convex subset of its upperspace and the upperspace is the union of these convex subsets. That these spaces are closed follows because the order is induced by a continuous map $F:DO(n)\rightarrow PO(n)$ and the order on $PO(n)$ also has closed upper- and downsets.

For the 7th statement we note that $(\sigma_1\otimes\sigma_2)^+ = \norm{\sigma_1\otimes \sigma_2} = \norm{\sigma_1}\norm{\sigma_2} = \sigma_1^+\sigma_2^+$. By first showing that $\rho\otimes \kappa \sleq \sigma \otimes \kappa$ iff $\rho\sleq \sigma$ for any state $\kappa$ it then easily follows.

The last point is true because for the normal positivity order we have $A\leq B$ implying $\lambda(A)\leq \lambda(B)$.
\end{proof}

When studying computation in a formal context such as when constructing semantics, the structure that is often required is that of a \emph{domain} \cite{scottbook2003}, which is a special type of order structure:
\begin{definition}
Let $(X,\leq)$ be a partially ordered set. $S\subseteq X$ is called \emph{upwards directed} when for all $x,y \in S$ we can find $z\in S$ such that $x,y\leq z$. $(X,\leq)$ is called \emph{directed complete} when any upwards directed set $S$ has a supremum $\vee S$. We say that $x\ll y$ ($x$ \emph{way below} $y$) when for any upwards directed set $S$ such that $y\leq \vee S$ we can find $s\in S$ such that $x\leq s$. We call $(X,\leq)$ a \emph{domain} when it is directed complete and the set of all elements way below $x$ is a directed set with supremum $x$ for all $x\in X$.
\end{definition}
\begin{theorem}
The QPE order is a domain on the state space $DO(n)$ for all $n\geq 1$. Furthermore we have
\begin{itemize}
    \item $(1-t)\rho + t\bot_n \ll \rho$ for any $0<t\leq 1$,
    \item When $\rho \ll \sigma$ then $(1-t)\sigma + t\rho \ll \sigma$ for $0<t \leq 1$,
    \item For all $\rho$ there exists a $\sigma$ such that $\rho \ll \sigma$ if and only if ker$(\rho) = \{0\}$.
\end{itemize}
\end{theorem}
\begin{proof}
Since the uppersets and downsets are closed subsets in the topology induced by the operator norm, by general consideration of a compact topological pospace, $\sleq$ is then directed complete and any increasing sequence will be convergent. It is then also enough to work with increasing sequences instead of arbitary directed sets (for more details see \cite{scottbook2003,martinphd}) and to prove it is a domain we only need to find for any $\rho$ an increasing sequence of approximations of $\rho$ that converges to $\rho$. Let $\rho(t) = (1-t)\rho + t\bot_n$. We will show that $\rho(t)\ll \rho$ so that this is indeed such a convergent sequence of approximations.

Let $(\sigma_i)$ be an increasing sequence of elements. This is then convergent in the operator norm to some $\sigma$. We suppose $\rho\sleq \sigma$. For every $0<t\leq 1$ we now need to find a $j$ such that $\rho(t)\sleq \sigma_j$. First note that $L^+(\sigma)\subseteq L^+(\sigma_j)\subseteq L^+(\sigma_i)$ when $i\leq j$. For some $N$ we must have equality for all $i\geq N$: suppose this is not the case, then there is a normalised $v$ such that $\sigma_i v = \sigma_i^+ v$ for all $i$ while $\sigma v \neq \sigma^+ v$. Then $\norm{\sigma_i v} = \sigma_i^+$. Since we have $\sigma_i\rightarrow \sigma$ we also have $\sigma_i^+ \rightarrow \sigma^+$, so $\norm{\sigma_i v} \rightarrow \sigma^+$, but $\norm{\sigma v} \leq \sigma^+ - \delta$ for some $\delta>0$. $(\sigma_i)$ converges in the matrix norm so this is a contradiction. We will now assume that $L^+(\sigma)=L^+(\sigma_j)$, otherwise we could just take the tail of the sequence where this is the case. 
Writing out the order inequality of $\rho(t)\sleq \sigma_j$ we get
\begin{equation*}
    (1-t)(\norm{\sigma_j}\rho - \norm{\rho}\sigma_j) + t(\norm{\sigma_j}\bot_n - \frac{1}{n}\sigma_j) \geq 0.
\end{equation*}

We note that $L^+(\sigma_j)=L^+(\sigma)\subseteq L^+(\rho(t))$ so that when we fill in a $v\in L^+(\sigma_j)$ in this expression it is exactly zero and we note that for any other $v\neq 0$ $v^\dagger(\norm{\sigma_j}\bot_n - \frac{1}{n} \sigma_j)v > 0$ is strictly bigger than zero. Now by adding and substracting $(1-t)(\norm{\sigma}\rho - \norm{\rho}\sigma)$ to the above expression we get
\begin{align*}
   & (1-t)\left((\norm{\sigma_j}-\norm{\sigma})\rho - \norm{\rho}(\sigma_j-\sigma)\right) \\
   +& (1-t)(\norm{\sigma}\rho - \norm{\rho}\sigma) + \frac{t}{n}(\norm{\sigma_j}I-\sigma_j) \geq 0
\end{align*}
The first term goes uniformly to zero when $\sigma_j\rightarrow \sigma$, the second term is nonnegative, and the last term is strictly positive for any $v$ not in $L^+(\sigma)$, so we see that for large enough $j$ this expression will be a positive operator so that indeed $\rho(t)\sleq \sigma_j$ for some $j$.

The role of $\bot_n$ can be replaced by any $\rho^\prime \ll \rho$ which proves the second statement.

For the third statement we note that when ker$(\rho)\neq\{0\}$ we can for every state $\sigma$ construct some sequence $a_j$ such that $a_j\rightarrow P$ where $P$ is the projection corresponding to $L^+(\sigma)$ while ker$(a_j)=\{0\}$. Since $\rho\sleq a_j$ would imply ker$(\rho)\subseteq $ ker$(a_j)$ this is not possible proving that $\rho$ doesn't approximate anything. Suppose now for the other direction that ker$(\rho)=\{0\}$. Then let $z(\lambda) = \lambda \rho + (1-\lambda)\bot_n$. Then there is a $\lambda >1 $ such that $z(\lambda)$ is still positive. The above argument then shows that $\rho\ll z(\lambda)$.
\end{proof}
To define semantics it is also required that the domain structure is monotone over all quantum channels\footnote{In the language of category theory: that the order structure is an enrichment of the category.}. Such an order structure is not possible, primarily because performing a partial trace operation can change a pure state to a completely mixed state.

\begin{lemma}Let $(DO(nm),\leq)$ for $n,m\geq 2$ be an ordered space such that each $\rho\in DO(nm)$ is below the pure state corresponding to its highest eigenvalue. Let $(DO(n), \leq^\prime)$ be an ordered set with the completely mixed state as the bottom element. Then the partial trace $\tr_2: M_{nm}\rightarrow M_n$ does not restrict to a monotone map $\tr_2: (DO(nm),\leq)\rightarrow (DO(n),\leq^\prime)$.
\end{lemma}
\begin{proof}
Let $P = \ket{M}\bra{M}$ be a maximally entangled state so that $\tr_2(P)=\bot_n$. Now let $\rho = tP + (1-t)A$ where $PA=AP=0$ and $t\geq \frac{1}{2}$. Then $L^+(\rho)=$ supp$(P)$ so that we must have $\rho\leq P$. If $\tr_2$ is monotone then we should now have $\tr_2(\rho)\leq^\prime \tr_2(P) = \bot_n$, but since $\bot_n$ is the least element in $DO(n)$ this is only possible when $\tr_2(\rho) = t\bot_n + (1-t)\tr_2(A) = \bot_n$. The only condition on $A$ was that it is orthogonal to $P$, so this won't be the case in general.
\end{proof}

Despite the order not being monotone over all channels, there are some non-pure channels that preserve the order. It is for instance preserved by depolarising noise.

\begin{definition}
    The \emph{depolarizing map} $D_t:M_n \rightarrow M_n$ is given by $D_t(\rho) = (1-t)\rho + t\bot_n$. When $t=1$ we will call it the \emph{completely depolarizing} map.
\end{definition}

\begin{lemma}
When $\kappa\sleq \rho\sleq \sigma$ we have $(1-t)\rho+t\kappa \sleq (1-t)\sigma + t\kappa$.
\end{lemma}
\begin{proof}
Note that $((1-t)\rho + t\kappa)^+ = (1-t)\rho^+ + t\kappa^+$. Writing out the order inequality we must show that
\begin{align*}
    &\left((1-t)\sigma^+ +t\kappa^+\right)((1-t)\rho + t\kappa) \\
    - &\left((1-t)\rho^+ +t\kappa^+\right)((1-t)\sigma + t\kappa) \geq 0
\end{align*}
which by cancelling terms can be written as
\begin{equation*}
    (1-t)^2(\sigma^+\rho - \rho^+\sigma) + t(1-t)\left((\sigma^+-\rho^+)\kappa + \kappa^+(\rho - \sigma)\right) \geq 0.
\end{equation*}
We note that the first term is positive because $\rho\sleq \sigma$, so if we can show that the second term is positive we are done. That is, we must show that
\begin{equation*}
    0 \leq (\rho-\sigma) + (\sigma^+ - \rho^+)\frac{\kappa}{\kappa^+}.
\end{equation*}
We have
\begin{equation*}
    0 \leq \sigma^+\rho-\rho^+\sigma = \rho^+(\rho-\sigma) + (\sigma^+-\rho^+)\rho
\end{equation*}
Now by dividing this expression by $\rho^+$ and adding the term $(\sigma^+ - \rho^+)(\frac{\kappa}{\kappa^+} - \frac{\rho}{\rho^+})$ to it (which is positive because $\sigma^+\geq \rho^+$ and $\kappa \sleq \rho$) we get the desired inequality.
\end{proof}
\begin{corollary}
$D_t$ is a monotone map for $\sleq$.
\end{corollary}
The order is also preserved by measurement results that preserve the highest eigenvector of a state, these are the measurements that can be considered `positive' as discussed in section 2.
\begin{lemma}
	Let $\rho\sleq \sigma$ and let $0\leq E\leq 1$ with $L^+(E)\cap L^+(\sigma) \neq \{0\}$, then
	$$\frac{E^\half \rho E^\half}{\tr(E\rho)} \sleq \frac{E^\half \sigma E^\half}{\tr(E\sigma)}.$$
\end{lemma}
\begin{proof}
	Follows in the same way as Theorem \ref{theor:evidence} and noting that if $A\geq B$ then for any Hermitian operator $E$ we have $EAE\geq EBE$.
\end{proof}
Note that this is not monotone over the entire state space, but only on that part on which the highest eigenvectors agree.
It is interesting that the monotonicity here holds using the standard projection postulate derived rule, and not using FLS Bayes' rule. This again points in the direction of there not being one definite `correct' Bayes' rule.

We could consider the QPE order as one in a family of order structures on states given by renormalising the states using some scalar:
\begin{definition}
Let  $f:DO(n)\rightarrow \mathbb{R}_{>0}$ be a continuous map. We define the $f$-renormalised order as
\begin{equation*}
    \rho \sleq^f \pi \iff \frac{\rho}{f(\rho)} \geq \frac{\pi}{f(\pi)} \iff f(\pi)\rho - f(\rho)\pi \geq 0.
\end{equation*}
\end{definition}
But it turns out that the QPE order is the unique renormalised order that preserves the convex structure of the state space.
\begin{theorem}
Let  $f:DO(n)\rightarrow \mathbb{R}_{>0}$ be continuous. Suppose $\sleq^f$ is an order where (1) the completely mixed state is the least element, (2) each state is below the pure state corresponding to its highest eigenvector and (3) $\sleq^f$ respects the convex structure of $DO(n)$. Then $\sleq^f = \sleq$.
\end{theorem}
\begin{proof}
Let $\rho = $ diag$(x)$ with $x^+=x_1\geq x_i$ for all $i$. Then $\rho$ must be below $P_1$, the projection to the first coordinate:
\begin{equation*}
    \rho\sleq^f P_1 \iff f(P_1)\rho -f(\rho)P \geq 0 \iff f(P_1)\rho^+ \geq f(\rho).
\end{equation*}
$\bot_n$ must be below $\rho$ giving
\begin{align*}
    \bot_n\sleq^f \rho &\iff f(\rho)\frac{1}{n}I_n - f(\bot_n)\rho \geq 0 \\
    &\iff f(\rho)\geq nf(\bot_n)\rho^+.
\end{align*}
So we see that if we can show that $f(\bot_n)=f(P_1)\frac{1}{n}$ we must have $f(\rho)\geq f(P_1)\rho^+$ which together with the inequality for $\rho\sleq^f P_1$ gives $f(\rho) = f(P_1)\rho^+$ so that we only need to show that $f(P_1)=f(P)$ for any pure state $P$ to finish the proof. Because we already have $f(\bot_n)\leq f(P_1)\bot_n^+ = f(P_1)\frac{1}{n}$ it suffices to show that it is not possible to have $f(\bot_n)< f(P_1)\frac{1}{n}$.

So assume that $f(\bot_n)< f(P_1)\frac{1}{n}$. Now let $\rho_\delta = $ diag$(x_\delta)$ given by $x=(\frac{1}{n}+\delta,\ldots,\frac{1}{n}+\delta,\frac{1}{n} - (n-1)\delta)$ for $0 < \delta < \frac{1}{n(n-1)}$. There is such a $\delta$ such that $f(\rho_\delta) \leq f(P_1)(\frac{1}{n}-(n-1)\delta)$ otherwise we would have $f(\bot_n) = f(\lim_{\delta\rightarrow 0} \rho_\delta) = \lim_{\delta\rightarrow 0}f(\rho_\delta) \geq \lim_{\delta\rightarrow 0} f(P_1)(\frac{1}{n}-(n-1)\delta) = f(P_1)\frac{1}{n}$ by continuity of $f$. Take a $\delta$ for which this inequality holds. Then we note that
\begin{align*}
    \rho_\delta \sleq^f P_n &\iff f(P_n)x_\delta - f(\rho_\delta)P_n \geq 0 \\
    &\iff f(P_n)(\frac{1}{n}-(n-1)\delta) \geq f(\rho_\delta)
\end{align*}
which is definitely the case when $f(P_n)\geq f(P_1)$. So assuming $f(P_n)\geq f(P_1)$ we have $\rho_\delta\sleq^f P_n$. Then because the order preserves the convex structure for $z(t) = t\rho_\delta + (1-t)P_n$ we should have $\rho_\delta\leq z(t)$ for all $0\leq t\leq 1$. Let $t=\frac{1}{1+n\delta}$ so that $z(t)=\bot_n$. Then $\rho_\delta\sleq^f \bot_n$ which by antisymmetry gives $\rho_\delta = \bot_n$ which is not the case. This means we must have $f(P_n)<f(P_1)$, but the choice of which coordinate we labeled $n$ and which as $1$ was arbitrary so we also get $f(P_1)<f(P_n)$. This is a contradiction, which means that we can't have $f(\bot_n)<f(P_1)\frac{1}{n}$. So now $f(\bot_n)= f(P_1) \frac{1}{n}$, but again the $1$ is arbitrary so that we must have $f(P_1)=f(P)$ for all pure states $P$ which means we have $f(\rho)=C \rho^+$ where $C=f(P)$. Dividing $f$ by  a positive scalar doesn't change the order so $\sleq^f = \sleq$.
\end{proof}

It is interesting to note that there is an infinite family of orders on states that respect convexity (see the author's Master's thesis \cite{weteringthesis}) and there are infinitely many orders that respect tensor products, e.g.\ renormalisations to functions $f(\rho) = \norm{\rho}_p^r$ where $\norm{\cdot}_p$ is a $p$-norm and $r\geq 1$ a real number, but the set of orders with both these properties seems to be really small: the only ones known are the QPE order and the two orders outlined in section 6.
 
\section{Quantum max-divergence}

The QPE order has some clear connections to purity of states as seen above. We might then wonder what the relation is with purity related quantities like von Neumann entropy or a relation such as majorisation. Unfortunately, they don't work nicely together. Denote the majorisation relation with $\prec$ and von Neumann entropy by $S$. If we take $x=$diag$(0.46,0.46,0.08)$ and $y=$diag$(0.7,0.2,0.1)$ then it is easily verifiable that $x\sleq y$, but not $x\prec y$ (or $y\prec x$) and if we take $x=\frac{1}{80}$diag$(30,29,11,10)$ and $y=\frac{1}{80}$diag$(34,23,12,11)$ then we have $x\sleq y$ while also $S(x)< S(y)$ (noting that entropy is supposed to be contravariant: it is bigger for more mixed quantities). 

A generalisation of the commonly used notions of Shannon entropy and Kullback-Leibler divergence is the concept of R\'enyi entropy/divergence. These quantities have been extended to quantum theory in two different forms: a straightforward generalisation \cite{mosonyi2011quantum,mosonyi2015quantum} and a `sandwiched' version \cite{muller2013quantum,quantumrenyi}. The latter seems to be the correct generalisation: it is monotone under application of quantum channels as you would expect and for $\alpha\geq 1$ they have an operational interpretation as generalised cut-off rates for state discrimination \cite{mosonyi2015quantum}. These quantities are defined as follows.

\begin{definition}
The \emph{quantum R\'enyi entropy and divergence of order} $\alpha \in (0,\infty)\setminus\{1\}$ for states $\sigma$ and $\rho$ are defined as
\begin{align*}
    H_\alpha(\sigma) &= \frac{1}{1-\alpha}\log\tr(\sigma^\alpha) = \frac{\alpha}{1-\alpha}\log\norm{\sigma}_\alpha \\
    D_\alpha(\sigma\lvert\rvert\rho) &= 
        \begin{cases}
        \frac{1}{\alpha-1}\log\tr\left[\left(\rho^{(1-\alpha)/2\alpha}\sigma\rho^{(1-\alpha)/2\alpha}\right)^\alpha\right] & \alpha<1\text{ or supp}(\sigma)\subseteq \text{supp}(\rho) \\
        \infty & \text{otherwise}
    \end{cases}
\end{align*}
where $\norm{\sigma}_p$ denotes the Schatten $p$-norm.
\end{definition}

There are a few special values of $\alpha$ that correspond to some better known quantities. When $\alpha\rightarrow 1$ we get back von Neumann entropy and the associated divergence. For $\alpha=2$ the entropy is equal to the negative logarithm of the Hilbert-Schmidt norm and for $\alpha=\half$ the divergence is related to the fidelity:
\begin{equation*}
    D_\frac{1}{2}(\sigma\lvert\rvert\rho) = -2\log\tr\left[\left(\rho^\frac{1}{2}\sigma\rho^\frac{1}{2}\right)^\frac{1}{2}\right] = -\log F^2(\sigma,\rho).
\end{equation*}
This is the only value of $\alpha$ for which the divergence is symmetric in its arguments.

For the limit of $\alpha\rightarrow \infty$ we get the \emph{quantum max-divergence}. Mosonyi and Ogowa showed that\cite{mosonyi2015quantum}:
\begin{align*}
    D_\infty(\sigma\lvert\rvert\rho) &= \lim_{\alpha\rightarrow\infty} D_\alpha(\sigma\lvert\rvert\rho) = \inf\{\gamma ~;~ \sigma\leq e^\gamma \rho\} \\
    &= \max_{M}\{D_\infty(M(\sigma)\lvert\rvert M(\rho))\}
\end{align*}
where the maximum is taken over all POVM's $M=(M_i)$ and $M(\sigma)$ denotes the probability distribution $(\tr\sigma M_i)$. The $D_\infty$ occurring in the bottom equation is the classical definition:
$$D_\infty(x\lvert \rvert y) = \log \max_i \frac{x_i}{y_i}.$$
By using the fact that the logarithm is monotonely increasing we could therefore also write the divergence as
\begin{equation*}
    D_\infty(\sigma\lvert\rvert\rho) = \log \max_M\max_i \frac{\tr \sigma M_i}{\tr \rho M_i} = \log \max_{0<M\leq 1}\frac{\tr \sigma M}{\tr \rho M}
\end{equation*}
which gives an interpretation of $D_\infty$ as measuring the worst case relative difference between states: given that an opponent can choose an arbitrary single measurement how well can he distuingish the states? For this reason we will also refer to quantum max-divergence as the \emph{worst case distuinguishability}. We will later on see that we can generalise the definition of the max-divergence to quantum channels and that we retrieve a similar sort of interpretation.

Analogous as shown for the related quantity of conditional min-entropy \cite{tomamichel2012framework} we will show that this max divergence is easy to calculate using the following result:
\begin{lemma}
Let $\sigma$ and $\rho$ be states, then
\begin{equation*}
    D_\infty(\sigma\lvert\rvert\rho)  = \begin{cases} \log\norm{\rho^{-\frac{1}{2}}\sigma \rho^{-\frac{1}{2}}} & \text{supp}(\sigma)\subseteq \text{ supp}(\rho) \\
    \infty & \text{otherwise}
    \end{cases}
\end{equation*}
where the norm is the operator norm and $A^{-1}$ is the generalised inverse that is zero on ker$(A)$.
\end{lemma}
\begin{proof}
We will assume that supp$(\sigma)\subseteq $ supp$(\rho)$, so that $D_\infty$ is finite. Let us first rewrite $D_\infty$ in a more workable form.
    \begin{align*}
        D_\infty(\sigma\lvert\rvert\rho) &= \inf\{\gamma ~;~ \sigma\leq e^\gamma \rho\} \\
        &= \inf\{\log\lambda ~;~ \sigma\leq \lambda \rho\} \\
        &= \log\inf\{\lambda ~;~ \sigma\leq \lambda \rho\}
    \end{align*}
    Let $P$ denote the projection onto the support of $\rho$. $\rho$ is a bijection when restricted to its support, and so is $\rho^{-\frac{1}{2}}$. We note that $A\mapsto \rho^{-\frac{1}{2}}A\rho^{-\frac{1}{2}}$ is a monotone map. This means that 
    \begin{equation*}
        \sigma\leq \lambda \rho \implies \rho^{-\frac{1}{2}}\sigma\rho^{-\frac{1}{2}}\leq \lambda \rho^{-\frac{1}{2}}\rho\rho^{-\frac{1}{2}} = \lambda P
    \end{equation*}
    The other direction also holds: assume that there is a $v$ such that $v^\dagger\sigma v > \lambda v^\dagger\rho v$. We may assume that $v$ is in the support of $\rho$. Because $\rho^{-\frac{1}{2}}$ is a bijection on its support there is a $w$ such that $\rho^{-\frac{1}{2}}w = v$. We then get $w^\dagger \rho^{-\frac{1}{2}}\sigma\rho^{-\frac{1}{2}} w > w^\dagger \lambda \rho^{-\frac{1}{2}}\rho\rho^{-\frac{1}{2}} w$, which shows the other implication.
    So now
    \begin{align*}
         D_\infty(\sigma\lvert\rvert\rho) &= \log\inf\{\lambda ~;~ \sigma\leq \lambda \rho\} \\
         &= \log\inf\{\lambda ~;~ \rho^{-\frac{1}{2}}\sigma\rho^{-\frac{1}{2}} \leq \lambda P\}.
    \end{align*}
    A similar argument as above shows that we may replace $P$ in that expression with $I_n$, the identity. We then note that $\norm{A} = \inf\{\lambda ~;~ \norm{Av} \leq \lambda \norm{v}\} = \inf\{\lambda ~;~ A \leq \lambda I_n\}$ where the last equality holds when $A$ is a positive operator. Applying this to the expression of $D_\infty$ above then proves the required statement.
\end{proof}

Written in this new form it becomes easy to prove the following:
\begin{theorem}
The quantum max-divergence is a quasi-metric: positive definite and satisfying the triangle inequality, but failing to be symmetric.
\end{theorem}
\begin{proof}
It was already shown in \cite{muller2013quantum} that all the quantum R\'enyi divergences are positive definite, but we will show it here explicitly.
To see that it is positive, we need to show that 
\begin{equation*}
    \norm{\rho^{-\frac{1}{2}}\sigma\rho^{-\frac{1}{2}}} \geq 1
\end{equation*}
whenever ker$(\rho)\subseteq$ ker$(\sigma)$ (otherwise the divergence will be infinite which is of course postive). Assume the converse and let $P$ be the projection onto the image of $\rho$. Then
\begin{align*}
    1 &= \tr(\sigma) = \tr(P\sigma P) = \tr(\rho^\frac{1}{2}\left(\rho^{-\frac{1}{2}}\sigma\rho^{-\frac{1}{2}}\right)\rho^\frac{1}{2}) \\
    &< \tr(\rho^\frac{1}{2} I \rho^\frac{1}{2}) = \tr(\rho) = 1
\end{align*}
which is a contradiction. 

If we have $D_\infty(\sigma\lvert\rvert \rho) = 0$ then $\norm{\rho^{-\frac{1}{2}}\sigma\rho^{-\frac{1}{2}}} = 1$ so that the above argument gives:
\begin{equation*}
    \tr(\rho^\frac{1}{2}\left(\rho^{-\frac{1}{2}}\sigma\rho^{-\frac{1}{2}}\right)\rho^\frac{1}{2}) = \tr(\rho^\frac{1}{2} I \rho^\frac{1}{2})
\end{equation*}
which is only possible when $\rho^{-\frac{1}{2}}\sigma\rho^{-\frac{1}{2}} = P$, the projection onto the image of $\rho$. This expression can be rewritten to $\sigma = \rho$, which proves positive definiteness.

For the triangle inequality, let $P=\rho^\frac{1}{2}\rho^{-\frac{1}{2}} = \rho^{-\frac{1}{2}}\rho^\frac{1}{2}$ be again the projection onto the image of $\rho$ and assume that all the supports are included in the right way (since otherwise both sides of the triangle inequality would be infinite), then
\begin{align*}
    \norm{\kappa^{-\frac{1}{2}}\sigma\kappa^{-\frac{1}{2}}} &\leq \norm{\kappa^{-\frac{1}{2}}P\sigma P\kappa^{-\frac{1}{2}}} \\
    &= \norm{\kappa^{-\frac{1}{2}}\rho^\frac{1}{2}\left(\rho^{-\frac{1}{2}}\sigma\rho^{-\frac{1}{2}}\right)\rho^\frac{1}{2}\kappa^{-\frac{1}{2}}}
\end{align*}
Note that this last expression has the form $\norm{ABA^\dagger}$, for which holds: $\norm{ABA^\dagger}\leq \norm{A}\norm{B}\norm{A^\dagger} = \norm{AA^\dagger}\norm{B}$ where in the last equality we used the C$^*$-algebra property of the norm. Applying this we get
\begin{equation*}
    \norm{\kappa^{-\frac{1}{2}}\sigma\kappa^{-\frac{1}{2}}} \leq \norm{\rho^{-\frac{1}{2}}\sigma\rho^{-\frac{1}{2}}}\norm{\kappa^{-\frac{1}{2}}\rho^\frac{1}{2}\rho^\frac{1}{2}\kappa^{-\frac{1}{2}}}
\end{equation*}
Noting that $\rho^\frac{1}{2}\rho^\frac{1}{2}=\rho$ and taking the logarithm on both sides then gives the triangle inequality for $D_\infty$.
\end{proof}

The positive-definiteness of this quantity was already known \cite{quantumrenyi}, this proof is merely a more direct version. Although not too hard to prove, the triangle inequality doesn't seem to have been described in the literature before.

A way to measure purity for a state would be to look at the maximum overlap it has with any pure state:
\begin{equation*}
    \max_v F^2(\rho,\ket{v}\bra{v}) = \max_v \bra{v} \rho \ket{v} = \norm{\rho}.
\end{equation*}
The min-entropy is minus the logarithm of this expression, so it is a measure of how far a state is from being pure. Now we can combine the quantum min-entropy, max-divergence and the QPE order in an intuitive way.

\begin{theorem}
\label{theor:divergence}
Let $\rho$ and $\sigma$ be states, then 
\begin{equation*}
    D_\infty(\sigma\lvert\rvert \rho) \geq H_\infty(\rho) - H_\infty(\sigma) = \log\frac{\norm{\sigma}}{\norm{\rho}}
\end{equation*}
and equality holds if and only if $\rho\sleq \sigma$:
\begin{equation*}
    \rho\sleq\sigma \iff D_\infty(\sigma\lvert\rvert \rho) = \log\frac{\norm{\sigma}}{\norm{\rho}}
\end{equation*}
\end{theorem}
\begin{proof}
We first prove the equality in the if direction. Suppose $\rho\sleq \sigma$. We then have $\norm{\sigma}\rho - \norm{\rho}\sigma \geq 0$. We conjugate this expression with $\rho^{-\frac{1}{2}}$ which will preserve positivity because $\rho$ is self-adjoint. The resulting expression is then
\begin{equation*}
    \norm{\sigma}P - \norm{\rho}\rho^{-\frac{1}{2}}\sigma\rho^{-\frac{1}{2}}\geq 0
\end{equation*}
where $P$ is the projection onto the image of $\rho$. But that means that for any normalised vector $v$ in the image of $\rho$ we have
\begin{equation*}
    \frac{\norm{\sigma}}{\norm{\rho}} \geq v^\dagger\left(\rho^{-\frac{1}{2}}\sigma\rho^{-\frac{1}{2}}\right)v
\end{equation*}
from which it follows that the lefthandside is bigger than the supremum over the $v$'s on the righthandside: $\frac{\norm{\sigma}}{\norm{\rho}} \geq \norm{\rho^{-\frac{1}{2}}\sigma\rho^{-\frac{1}{2}}}$. To see that this is actually an equality note that if we pick a $v$ such that $\sigma v = \sigma^+ v$, that we then also have $\rho v = \rho^+ v$. By taking the logarithms on both sides we are done.

The other direction of the equality follows by proving the inequality. To show this inequality we will assume the inequality in the other direction and show that in that case we must actually have equality. So suppose we have $D_\infty(\sigma\lvert\rvert \rho) \leq H_\infty(\rho)-H_\infty(\sigma)$. Then in particular the divergence is finite which means that the kernel of $\rho$ is included in the kernel of $\sigma$, or equivalently that the image of $\sigma$ is included in the image of $\rho$. We now have for any normalised $v$:
\begin{equation*}
    \frac{\norm{\sigma}}{\norm{\rho}} \geq v^\dagger\left(\rho^{-\frac{1}{2}}\sigma\rho^{-\frac{1}{2}}\right)v
\end{equation*}
but we then immediately see that if $P$ is the projector on the image of $\rho$ then
\begin{equation*}
    \frac{\norm{\sigma}}{\norm{\rho}}P \geq \rho^{-\frac{1}{2}}\sigma\rho^{-\frac{1}{2}}.
\end{equation*}
Conjugating both sides by $\rho^\frac{1}{2}$ preserves positivity and this gives us
\begin{equation*}
    \frac{\norm{\sigma}}{\norm{\rho}}\rho \geq P\sigma P = \sigma
\end{equation*}
where the last inequality follows because the image of $\sigma$ is contained in the image of $\rho$. By multiplying both sides by $\norm{\rho}$ we get $\rho\sleq \sigma$ for which we know that actually $D_\infty(\sigma\lvert\rvert \rho) = H_\infty(\rho)-H_\infty(\sigma)$. So we indeed have $D_\infty(\sigma\lvert\rvert \rho) \geq H_\infty(\rho)-H_\infty(\sigma)$.
\end{proof}

This statement can be rewritten to
$$H_\infty(\sigma) = H_\infty(\rho) - D_\infty(\sigma\lvert\rvert \rho) ~\text{when }\rho\sleq\sigma.$$
Written in this way it is clear that this theorem generalises a common way to define an entropy quantity in terms of a divergence:
\begin{equation*}
    H(\sigma) = H(\bot_n) - D(\sigma\lvert\rvert \bot_n).
\end{equation*}
We could therefore interpret the $\sleq$ relation as telling us when a state $\rho$ is akin to the completely mixed state for a $\sigma$ with respect to $D_\infty$.

 Recalling that $D_\infty$ measures the worst case measurable difference between $\sigma$ and $\rho$ we also get an operational interpretation of the QPE order:

\noindent\hspace{1em}\begin{minipage}[h!]{\columnwidth-1em}
\vspace{0.5em}
\emph{Suppose we want to construct a state $\sigma$, but we can only achieve a maximum purity of $M\geq H_\infty(\sigma)$ as measured in terms of min-entropy. What state $\rho$ should we construct to minimise the worst case difference between $\sigma$ and $\rho$ as measured in terms of quantum max-divergence?}

\emph{Answer}: Pick a $\rho$ such that $H_\infty(\rho)=M$ and $\rho\sleq \sigma$.
\vspace{0.5em}
\end{minipage}

Section 2 told us that $\rho\sleq \sigma$ if and only if $\sigma$ is obtained from $\rho$ by some positive Bayesian evidence. The above theorem now implies that a Bayesian update changes the state as minimally as possible: any other state $\sigma^\prime$ with the same purity as $\sigma$ will have a max-divergence $D_\infty(\sigma^\prime\lvert\rvert \rho)$ at least as large as $D_\infty(\sigma\lvert\rvert \rho)$ when $\sigma$ is obtained from $\rho$ by a Bayesian update. 

There is also another corollary to Theorem \ref{theor:divergence}:

\begin{corollary}
    $\rho_1\otimes\rho_2\sleq \sigma_1\otimes \sigma_2$ if and only if $\rho_i\sleq \sigma_i$ for $i=1,2$.
    In particular, for all $k\in\mathbb{N}_{>0}$: $\rho^{\otimes k} \sleq \sigma^{\otimes k}$ if and only if $\rho\sleq \sigma$.
\end{corollary}
\begin{proof}
    The if direction should be clear. The only if direction follows because $H_\infty(\rho_1\otimes \rho_2) = H_\infty(\rho_1) + H_\infty(\rho_2)$ and $D_\infty(\sigma_1\otimes \sigma_2 \lvert\rvert\rho_1\otimes \rho_2) = D_\infty(\sigma_1\lvert\rvert\rho_1) + D_\infty(\sigma_2\lvert\rvert\rho_2)$.
\end{proof}
This last fact means that nothing extra is gained by considering the states as part of a bigger system: no catalysis is possible as preparing additional copies of a state doesn't let it suddenly become below another state. This is in contrast to some other measures of convertability such as asymptotic LOCC where considering additional copies of a state allows for additional possible conversions.

\section{Ordering quantum channels}
There are a couple of ways we could try to lift this order on states to an order on Completely Positive Trace Preserving (CPTP) maps. Some desirable properties we would like are that the completely depolarising map $D_1$ is the bottom element, that the unitary evolutions are maximal elements, and that when we tensor maps together that the order is preserved. In this section let $\Phi,\Psi: M_n\rightarrow M_k$ be a couple of CPTP maps.

Perhaps the most straightforward way is to define the order analogously to how we defined it on states: $\Phi\sleq \Psi$ iff $\norm{\Psi}\Phi - \norm{\Phi}\Psi \geq 0$. The problem with this is that $\norm{\Psi} = \norm{\Psi(1)}$ so that any unital map has norm equal to 1. This means in particular that this order is discrete when considering unital maps.

Another way to make an order is to define $\Phi\sleq \Psi$ iff $\Phi(\rho)\sleq\Psi(\rho)$ for all states $\rho$. This order has $D_1$ as the minimal element and the unitary conjugations are maximal, but this order is not preserved by tensoring maps together. To see this let $D_1,id_2:M_2\rightarrow M_2$ be respectively the completely depolarising qubit channel, and the identity. Then we obviously have $D_1\sleq id_2$ in this order, but we don't have $D_1\otimes id_2 \sleq id_2\otimes id_2$ which can be seen by plugging in any operator for which the eigenvector corresponding to the highest eigenvalue is not separable and for which the partial trace still has distinct eigenvalues.

We could try to solve the problem of the previous order by defining a `completely positive' variant: $\Phi\sleq \Psi$ iff $(\Phi\otimes id_l)(\rho)\sleq (\Psi\otimes id_l)(\rho)$ for all states $\rho$ in $M_n\otimes M_l$ for all $l$, but it seems like this order is discrete (as seen by the example with the completely depolarising map above).

A construction that works is to transform the maps to states using the Choi-Jamio\l{}kowski isomorphism.

\begin{definition}
\emph{Channel-state duality / Choi-Jamio\l{}kowski isomorphism} \cite{jiang2013channel}: To each CPTP map $\Phi:M_n\rightarrow M_k$ we associate a state $J(\Phi)$ in $M_n\otimes M_k$ in the following way:
\begin{equation*}
    J(\Phi) = \frac{1}{n}(id_n\otimes \Phi)(\ket{M}\bra{M}) = \frac{1}{n}\sum_{i,j} E_{ij}\otimes \Phi(E_{ij})
\end{equation*}
where $\ket{M}$ is the maximally entangled state and $E_{ij}$ is the matrix with a $1$ at the $ij$ position and zero's everywhere else. The factor $\frac{1}{n}$ ensures that $\tr J(\Phi) = 1$ when $\Phi$ is trace preserving.
\end{definition}
\begin{lemma}
Let $\Phi: M_n\rightarrow M_k$ be a linear map.
\begin{enumerate}
\item $J(\Phi)$ is positive definite if and only if $\Phi$ is completely positive and has normalised trace iff $\Phi$ is trace preserving.
\item $\Phi(\rho) = n\tr_1(J(\Phi)^{T_1}(\rho\otimes I_n))$ where $\tr_1$ denotes the partial trace on the first system and $T_1$ denotes the partial transpose on the first system.
\item The matrix rank of $J(\Phi)$ is equal to the Krauss rank of $\Phi$.
\item $J(\Phi_1\otimes \Phi_2)$ is equal to $J(\Phi_1)\otimes J(\Phi_2)$ up to some unitary permutation operation that doesn't depend on the $\Phi_i$.
\item $J((1-t)\Phi_1 + t\Phi_2) = (1-t)J(\Phi_1) + tJ(\Phi_2)$.
\item Let $\Xi_1: M_n\rightarrow M_n$ and $\Xi_2:M_k\rightarrow M_l$ be CPTP maps. Then $J(\Xi_2\circ \Phi\circ \Xi_1) = (\Xi_1^T\otimes \Xi_2)(J(\Phi))$ where $\Phi^T$ denotes the transpose of $\Phi$.
\end{enumerate}
\end{lemma}

\begin{definition}
Let CPTP$[M_n,M_k]$ be the convex space of CPTP maps from $M_n$ to $M_k$. We define the \emph{channel QPE order} on it in the following way for $\Phi,\Psi \in $ CPTP$[M_n,M_k]$:
\begin{equation*}
    \Phi\sleq \Psi \iff J(\Phi)\sleq J(\Psi)
\end{equation*}
where the righthandside is the QPE order on states.
\end{definition}
Note that the channel QPE order doesn't depend on the specific choice of the maximally entangled state because all maximally entangled states are equivalent up to unitary conjugation which the QPE order is invariant under. We also get the following equivalent formulation:
\begin{align*}
	\Phi\sleq\Psi &\iff J(\Phi)\sleq J(\Psi) \iff \frac{J(\Phi)}{\norm{J(\Phi)}}\geq \frac{J(\Psi)}{\norm{J(\Psi)}} \\
	&\iff (I_n\otimes (\frac{\Phi}{\norm{J(\Phi)}} - \frac{\Psi}{\norm{J(\Psi)}}))(\ket{M}\bra{M}) \geq 0 \\
	&\iff \frac{\Phi}{\norm{J(\Phi)}} - \frac{\Psi}{\norm{J(\Psi)}} \geq^{cp} 0
\end{align*}
where with $\Xi\geq^{cp} 0 $ we mean that $\Xi$ has to be completely positive. So this order can also be seen as a simple rescaling of the standard positivity ordering on maps. While the standard order is discrete when restricting to unital channels, this modified order has a lot more structure. Using the properties of the Jamio\l{}kowski isomorphism and the QPE order we easily get the following.
\begin{lemma}
The channel QPE order on CPTP$[M_n,M_k]$ has the following properties:
\begin{enumerate}
\item The completely depolarizing map $D_1(\rho) = \bot_k$ is the bottom element (below all maps).
\item The isometries (Krauss rank 1 operators) are maximal.
\item The order respects the convex structure of CPTP$[M_n,M_k]$ and is preserved by tensor products.
\item The order is directed complete and is a domain.
\end{enumerate}
\end{lemma}

The quantum max-divergence between the Jamio\l{}kowski states of some channels has an operational interpretation as the worst case max-divergence when applying the channels to arbitrary (possibly entangled) states:
\begin{theorem}
Let $\Psi,\Phi: M_n\rightarrow M_k$ be CPTP maps. Then
\begin{align*}
    \max_{m\geq 1}\max_{\rho\in DO(mn)} D_\infty((I_m\otimes\Psi)(\rho)\lvert\rvert(I_m\otimes\Phi)(\rho)) =  D_\infty(J(\Psi)\lvert\rvert J(\Phi))
\end{align*}
where the maximums are taken respectively over all natural numbers $m\geq 1$ and bipartite states $\rho\in M_m\otimes M_n$.
\end{theorem}
\begin{proof}
	We will use the equality $D_\infty(\sigma\lvert\rvert\rho) = \log\inf\{\lambda ~;~ \sigma\leq \lambda\rho\}$. We can then write
	\begin{align*}
		&\max_m\max_\rho D_\infty((I_m\otimes\Psi)(\rho)\lvert\rvert(I_m\otimes\Phi)(\rho))  \\
		=& \max_m\max_\rho\log\inf\{\lambda ~;~ (I_m\otimes\Psi)(\rho)\leq \lambda(I_m\otimes \Phi)(\rho)\} \\
		=& \log\inf\{\lambda ~;~ \forall m\geq 1,\forall \rho\in DO(mn): (I_m\otimes\Psi)(\rho)\leq \lambda(I_m\otimes \Phi)(\rho)\} \\
		=& \log\inf\{\lambda ~;~ \forall m\geq 1: I_m\otimes \Psi \leq \lambda I_m\otimes \Phi\} \\
		=& \log\inf\{\lambda ~;~ \Psi \leq^{cp} \lambda\Phi\} \\
		=& \log\inf\{\lambda ~;~ J(\Psi)\leq \lambda J(\Phi)\} \\
		=& D_\infty(J(\Psi)\lvert\rvert J(\Phi))
	\end{align*}
\end{proof}

When $\Phi\sleq \Psi$ we know that $D_\infty(J(\Psi)\lvert\rvert J(\Phi)) = \log\frac{\norm{J(\Psi)}}{\norm{J(\Phi)}}$ so that by this above equality the worst case distinguishability when applying these channels to states is minimised. This shows an operational interpretation of the channel QPE order: If the goal is to implement a channel $\Psi$, but the desired level of purity can't be achieved, then you should strive to instead implement a channel that is below $\Psi$ in the channel QPE order and is as pure as possible.

The order also seems to be related to the problem of entanglement distribution where Alice tries to create a maximally entangled state between her and Bob using a non-pure channel. This problem can be formalised in the following way: Find the maximal possible fidelity $F^2(\ket{M}\bra{M}, (I_n\otimes \Phi)(\rho))$ for any maximally entangled state $\ket{M}$ and bipartite state $\rho$. As shown in \cite{verstraete2002quantum} this quantity is maximised when $\rho$ is the eigenvector corresponding to the maximum eigenvalue of $J(\Phi)$ and then it is equal to that highest eigenvalue. 

A necessary condition for $\Phi\sleq\Psi$ to hold is that $\norm{J(\Phi)}\leq \norm{J(\Psi)}$ and that this norm is achieved on the same vector, so this means that $\Psi$ is strictly better at distributing entanglement.

\section{Other orders on states}
There are other orders on the space of quantum states that share many of the properties of the QPE order. The characteristic properties that we will consider are unitary conjugation invariance, preservation of convex structure, having the completely mixed state as bottom element and the pure states as maximal states, and directed completeness. The orders below here all satisfy these properties. Additional properties that the orders below don't necessarily have are being preserved by tensor product, having closed upper- and downsets, being a domain and respecting the kernel of states: $\rho\leq \sigma \implies $ ker$(\rho)\subseteq $ ker$(\sigma)$. Note that for $n=2$ all the orders below coincide with the QPE order, which isn't too surprising as it was shown \cite{martin2008domain} that this order is the unique order on $DO(2)$ with the maximally mixed state as the least element and that respects the convex structure of $DO(2)$ such that all unital qubit channels are Scott continuous.

The \emph{spectral order} $\sleq^s$ by Coecke and Martin \cite{coecke2010partial} satisfies all the conditions above except that it isn't preserved by tensoring and it is not a domain (for $n>2$). Note that the paper introducing the spectral order proves the domain property on a mistaken assumption. A concrete counter-example is given in the author's Master thesis \cite{weteringthesis}. The spectral order is an example of what in \cite{wetering2016entailment} is called a \emph{restricted information order}. All these orders are not domains and aren't preserved by tensoring.

The order below appeared in \cite{wetering2016entailment} and \cite{weteringthesis}.
\begin{definition}
Let $\rho^-$ denote the least nonzero eigenvalue of $\rho$ and $L^-(\rho)$ the corresponding eigenspace. We define the \emph{least eigenvalue order} as follows: $\rho\sleq^- \sigma$ if and only if one of the following holds:
\begin{enumerate}
    \item ker$(\rho) = $ ker$(\sigma)$ and $\rho^-\sigma - \sigma^-\rho \geq 0$.
    \item ker$(\rho)\subset $ ker$(\sigma)$ and $L^-(\rho)\cap $ ker$(\sigma) \neq \{0\}$.
\end{enumerate}
\end{definition}

It does not have closed downsets and is not a domain (but is directed complete). The intuition behind this order is that while the QPE order is a renormalisation to the highest eigenvalue, here we `normalise' to the \emph{lowest} eigenvalue (taking special care to deal with the zero's). It is interesting to note that for $n=3$ the intersection of this order with the QPE order is equal to the spectral order.

There is one more known relevant partial order structure on states. It satisfies all the properties outlined in the beginning of the section, but for respecting the kernel of states. So in particular it \emph{is} closed, a domain and preserved by tensors (this is proven in roughly the same way as for the QPE order in Theorem \ref{theor:properties}).
\begin{definition}
	Let $\rho$ and $\sigma$ be states in $M_n$ and $I_n$ the identity. Then $\rho\sleq^\prime \sigma$ if and only if
	$$\rho^+ I_n - \rho \leq \sigma^+ I_n - \sigma \iff \sigma - \rho \leq (\sigma^+ - \rho^+)I_n.$$
\end{definition}
This order is a generalisation of the solution to a classical problem: 

Suppose we can win a 100 dollars by betting on one of $n$ boxes and our knowledge of which one is correct is given by a probability distribution $x$. We would of course pick the box to which we assign the highest probability for an expected profit of $100x^+$. Now someone offers us $A_i$ money if instead we pick box $i$. When should we accept this offer? The expected profit would be $A_i + 100x_i$ and this needs to be bigger than $100x^+$, so we should accept when $A_i/100 \geq x^+ - x_i$. Now we can say that someone with knowledge $y$ is more certain when in every case he would need more money (higher $A_i$) to change his beliefs. So $y$ is more certain than $x$ when $x^+ - x_i \leq y^+ - y_i$ for all $i$. This is precisely the above definition of $\sleq^\prime$.

We have $\rho \sleq^s \sigma \implies \rho \sleq \sigma \implies \rho \sleq^\prime \sigma$.

The orders $\sleq$, $\sleq^-$ and $\sleq^\prime$ are the only orders on states known that both preserve the convex and tensor product structure of the state space. Of these three the QPE order is special in that it is the only one preserving the kernel of states: when $\sigma\sleq\rho$ we have ker$(\rho)\subseteq$ ker$(\sigma)$. This does not hold for the other two orders. It is currently not known if the QPE order is the unique order having these properties (preserving the kernel, convex structure and tensor product).

\emph{Acknowledgements}: This work is supported by the ERC under the European Union’s Seventh Framework Programme (FP7/2007-2013) / ERC grant n$^\text{o}$ 320571. The author would like to thank Aleks Kissinger for valuable comments and pointing out important references.

\bibliographystyle{plain}
\bibliography{cite}

\end{document}